\documentclass[a4paper]{elsarticle}
\usepackage[top=0.75in,bottom=0.75in,left=1.0in,right=1.0in]{geometry}
\usepackage{graphicx,subcaption,amsthm,amsmath,amssymb,mathtools,bm}
\usepackage[bitstream-charter]{mathdesign}
\usepackage[T1]{fontenc}
\usepackage{textcomp}
\usepackage{hyperref}
\hypersetup{colorlinks=true,
citecolor=blue,linkcolor=blue
}
\usepackage{url}

\newtheorem{theorem}{Theorem}
\newtheorem{lemma}{Lemma}

\newtheorem{remark}{Remark}

\newcommand\norm[1]{\left\lVert#1\right\rVert}

\makeatletter
\def\ps@pprintTitle{%
 \let\@oddhead\@empty
 \let\@evenhead\@empty
 \def\@oddfoot{\centerline{\thepage}}%
 \let\@evenfoot\@oddfoot}
\makeatother

\begin{document}

\begin{frontmatter}
\title{Attitude Control of a Novel Tailsitter: Swiveling Biplane-Quadrotor}

\author{Nidhish Raj\fnref{fn1}}
\ead{nraj@iitk.ac.in}
\author{Ravi N Banavar\fnref{fn2}}
\ead{banavar@iitb.ac.in}
\author{Abhishek\fnref{fn3}}
\ead{abhish@iitk.ac.in}
\author{Mangal Kothari\fnref{fn4}}
\ead{mangal@iitk.ac.in}

\fntext[fn1]{Doctoral Student, IIT Kanpur}
\fntext[fn2]{Professor, IIT Bombay}
\fntext[fn3]{Associate Professor, IIT Kanpur}
\fntext[fn4]{Associate Professor, IIT Kanpur}

\address{Department of Aerospace Engineering,
IIT Kanpur, Kanpur, UP, India 208016}

\begin{abstract}
This paper proposes a solution to the attitude tracking problem for a novel quadrotor tailsitter unmanned aerial vehicle called swiveling biplane quadrotor. The proposed vehicle design addresses the lack of yaw control authority in conventional biplane quadrotor tailsitters by proposing a new design wherein two wings with two attached propellers are joined together with a rod through a swivel mechanism. The yaw torque is generated by relative rotation of the thrust vector of each wing. The unique design of this configuration having two rigid bodies interconnected through a rod with zero torsional rigidity makes the vehicle underactuated in the attitude configuration manifold. An output tracking problem is posed which results in a single equivalent rigid body attitude tracking problem with second order moment dynamics. The proposed controller is uniformly valid for all attitudes and is based on dynamic feedback linearization in a geometric control framework. Almost-global asymptotic stability of the desired equilibrium of the tracking error dynamics is shown. The efficacy of the controller is shown with numerical simulation and flight tests.

\end{abstract} 
\end{frontmatter}


\section{Introduction}

Vertical take-off and landing capable hybrid vehicles (VTOL) are gaining interest due to its unique ability to hover like a rotorcraft and fly efficiently long distances like a fixed-wing aircraft. They are envisaged to play the role of air-taxi for transportation in urban environment, emergency first-aid/medical supply vehicle, commercial package delivery drone etc. to name a few. Typical configuration of such vehicles include tilt-rotor, tilt-wing and tailsitter configurations. Unlike the remaining configurations, the tailsitter has the advantage of having no tilting mechanism which simplifies the mechanical complexity and its associated failure points. The vehicle transitions between the two operating modes ---  hover mode and cruise flight mode --- by performing a maneuver which reorients the entire vehicle by almost 90 deg. An attitude tracking controller which is valid for all orientations is critical towards achieving this goal.

The most popular version of a tailsitter is the flying wing with two propellers \cite{verling2016full, ritz2017global}. In the hover mode, the rolling moment is generated by differential thrust of motors while pitching and yawing moments are produced by the flaps operating in the propeller downwash. The control authority of the vehicle is relatively low in pitch axis as the moment arm of the aerodynamic force generated by the flaps is small. On the other hand, a  quadrotor tailsitter, such as monoplane type \cite{zhou2017unified} or biplane type \cite{hrishikeshavan2014design, hrishikeshavan2013control, swarnkar2018biplane, chipade2018systematic} (Fig. \ref{fig:biplane_axis}), generates rolling and pitching moment with differential thrust and therefore has good control authority about these axes. However, it has relatively low control authority in yaw as explained in Sec. \ref{sec:design_mot}. 

Previous work on rigid tailsitter vehicles have addressed several important problems pertaining to autonomous flight, namely attitude control  \cite{verling2016full, ritz2018global}, robust transition control \cite{naldi2013robust}, trajectory tracking control \cite{ritz2017global} and design of optimal transition maneuvers \cite{naldi2011optimal, verling2017model, maqsood2012optimization}. Verling et al. \cite{verling2016full} focused on developing an attitude control on $SO(3)$ for a flying wing tailsitter which is valid for all attitude configurations and operating modes. It also develops a model for mapping actuator inputs (flap deflection and motor speed) to forces and moments using wind tunnel data. In \cite{verling2017model}, the work was extended to perform optimal back transition (from cruise to hover) with a cost function penalizing the altitude deviation during the transition and the control input being pitch angle and thrust input. Ritz and D'Andrea \cite{ritz2018global} have developed a very innovative attitude controller for the flying wing tailsitter by solving for an optimal reorientation maneuver which minimizes the torque requirement about the weakest axes (pitch and yaw). The resulting optimal solutions for a rich set of initial conditions were stored using lookup table and the controller was validated using extensive experimentation. In \cite{ritz2017global}, they extended their work for trajectory tracking using a cascaded control architecture and provided experimental validation. The aerodynamic forces and moments were captured into a heuristic model, the parameters of which were estimated using online learning scheme. Zhuo et al. \citep{zhou2017unified} proposed a unified controller for hover, cruise and transition regimes by obtaining the desired thrust and attitude using a nonlinear solver. The work was done for a quadrotor monoplane tailsitter and validated in simulation. 
Naldi and Marconi \cite{naldi2011optimal} consider the problem of generating minimum time and energy optimal transition trajectories for tailsitters using numerical techniques. In \cite{naldi2013robust} they extend their work to design  a robust feedback controller for maintaining the vehicle withing a safe flight envelope in the presence of wind during the transition maneuver.

A novel tailsitter called the swiveling biplane-quadrotor (Fig. \ref{fig:biplane_exploded} and \ref{fig:biplane_model}) is introduced in this paper. The vehicle changes its shape by twisting the two wings about $e_X$-axis to produce control moment about $e_Z$-axis. The proposed vehicle is an improvement over the biplane-quadrotor (Fig. \ref{fig:biplane_axis}), first introduced by Hrishikeshavan et al. \citep{hrishikeshavan2014design}, in terms of aerodynamic efficiency and the ability to generate large control torque about body frame Z-axis.

\subsection{Design Motivation} \label{sec:design_mot}
A quadrotor produces the net thrust $T$ and body frame control torques $(M_x,M_y,M_z)$ by varying the speed of the four rotors ($\omega_i, i = 1..4$) and obeys the following relation (valid for near hover condition)
\begin{equation*}
\begin{bmatrix}
T \\
M_x \\
M_y \\
M_z
\end{bmatrix} = \begin{bmatrix}
k_f	& k_f	& k_f	& k_f \\
-k_f l & k_f l & k_f l & -k_f l \\
k_f l & -k_f l & k_f l & -k_f l \\
k_\tau & k_\tau & -k_\tau & -k_\tau
\end{bmatrix} \begin{bmatrix}
\omega_1^2 \\
\omega_2^2 \\
\omega_3^2 \\
\omega_4^2
\end{bmatrix}, \quad \text{where} \quad k_f l : k_\tau \sim 10:1,
\end{equation*}
$l$ is half length of motor separation, and $k_f$ and $k_\tau$ are respectively the thrust and torque coefficients of the rotor. The large difference in order of magnitude of thrust and torque coefficients in the above equation imply that the magnitude of torque generated about the body frame Z-axis, $M_z$, is about 10 times less that the ones produced about the roll and pitch axes for the same differential speed of rotors. 

\begin{figure}[!htbp]
\begin{center}
\includegraphics[width=0.5\linewidth]{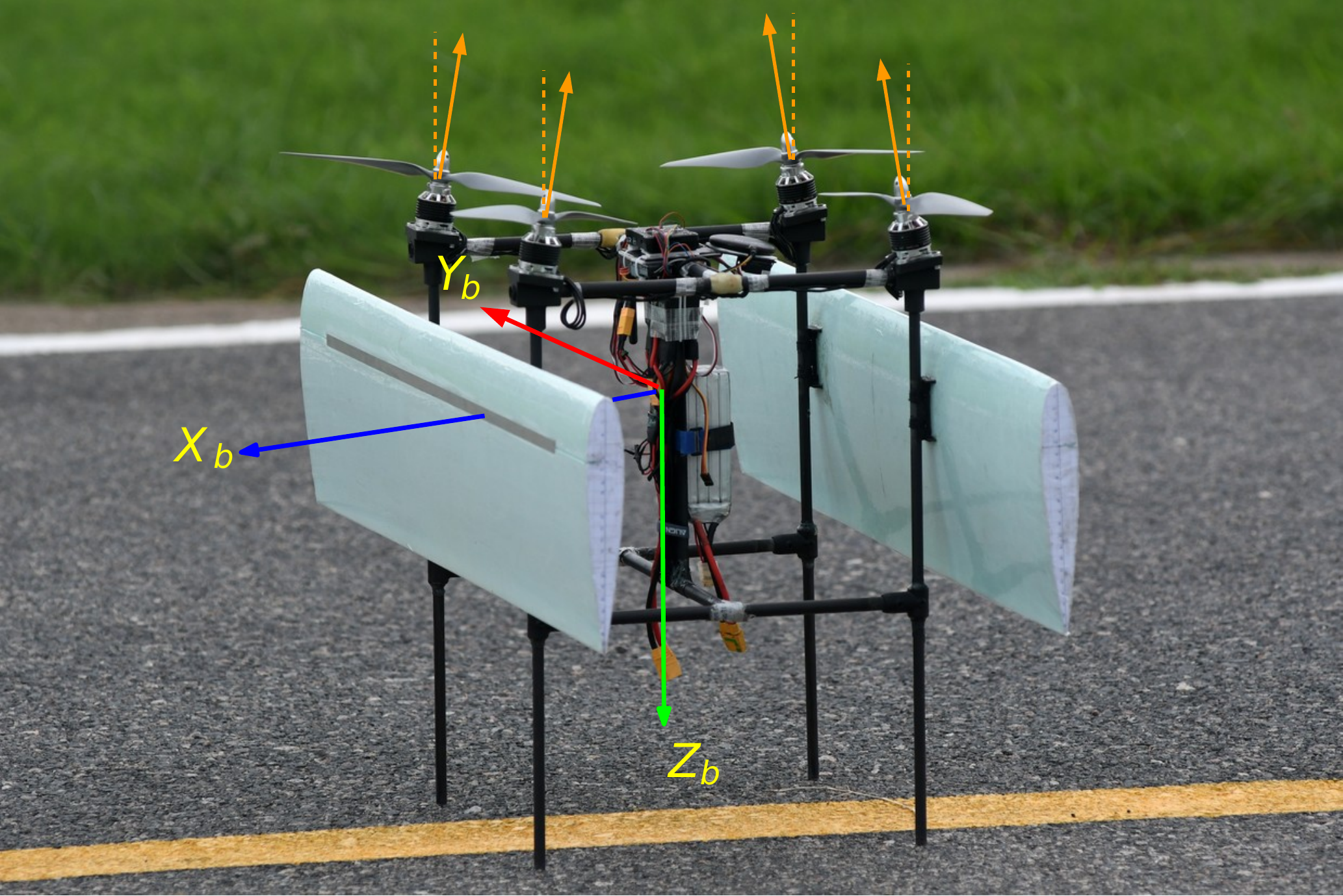}
\caption{Conventional biplane-quadrotor developed at IIT Kanpur with inward tilted motors for augmenting yaw control authority.}
\label{fig:biplane_axis}
\end{center}
\end{figure}

A conventional biplane-quadrotor is a regular quadrotor with two wings rigidly attached to it as shown in Fig. \ref{fig:biplane_axis}. It takes-off and lands like a quadrotor and transitions to/from level cruise flight (biplane mode) by pitching down/up about the $Y_b$-axis by 90 deg. The moment of inertia of the biplane-quadrotor about $Z_b$-axis is high and it experiences large aerodynamic damping moment about the body frame Z-axis. Further, the aerodynamic damping moment due to wing increases with forward speed in the biplane mode. Thus, high inertia, large aerodynamic damping added with low torque generation capability about $Z_b$ axis makes it less controllable about the axis. To improve the available Z-axis torque in the original design, the motors were permanently tilted inwards by 10 degree \cite{hrishikeshavan2013control} (see Fig. \ref{fig:biplane_axis}) and ailerons were introduced to augment the motor torques. However, the tilted motors reduce the hover efficiency and the use of ailerons make forward flight less efficient and require use of two additional actuators, increasing the complexity and reliability of this design. 

The swiveling biplane-quadrotor design does away with these two features and generates the Z-axis torque by tilting the two wings in opposite direction about body frame X-axis as shown in Fig. \ref{fig:biplane_model}. In order to realize this, the two wings are joined by a slender rod through a bearing, and they are tilted using the differential thrust of motors. This is different from other  methods used for morphing such as the ones described in \cite{ruben2017transformer, falanga2019foldable} where dedicated servo motors were used. The parameters of the swiveling biplane-quadrotor reveal that tilting the two wings by 15 degrees produces four times the maximum Z-axis torque attained by the previous design. Moreover, in the biplane-quadrotor design, producing a large Z-axis torque would make two diagonally-opposite motors operate at very low RPM, which in turn would make the vehicle less stable about the roll and pitch axes. In addition to the torque advantage, this design reduces the clutter between the wings with a single rod connection and makes the vehicle aerodynamically cleaner, unlike the conventional biplane-quadrotor. Further, building a conventional quadrotor by connecting the two wings with a slender rod without a bearing lowers the resonant frequency of the swiveling flexible mode due to low torsional stiffness to moment of inertia ratio (see attached video). This risks the structural integrity of the vehicle and lowers the control bandwidth. Introduction of the swiveling bearing makes the torsional stiffness zero thereby eliminating the resonance issues.

\subsection{Contribution}
The proposed vehicle has four degrees of freedom in the rotational configuration (3 DOF for one wing, 1 DOF for the relative angle between the wings), out of which only three are directly actuated (roll, pitch and swiveling motion). The yaw motion is not directly actuated by the motors which leads to an underactuated attitude control problem. The main contribution of this work, apart from introducing the novel vehicle, is in defining an attitude tracking control problem for the proposed vehicle and providing a practical solution.

The paper is organized as follows: First, the rotational dynamics of the vehicle is derived in Sec. \ref{sec:model} and the attitude tracking control problem and a solution to it is proposed in Sec. \ref{sec:att_control}. Next, the simulation study of the proposed controller under certain uncertainties in parameters is provided in Sec. \ref{sec:sim}. Finally, the proposed controller is  experimentally validated and its procedure and results are given in Sec. \ref{sec:exp}. 


\section{Vehicle Attitude Dynamics} \label{sec:model}

\begin{figure}[!htbp]
\begin{center}
\includegraphics[width=0.5\linewidth]{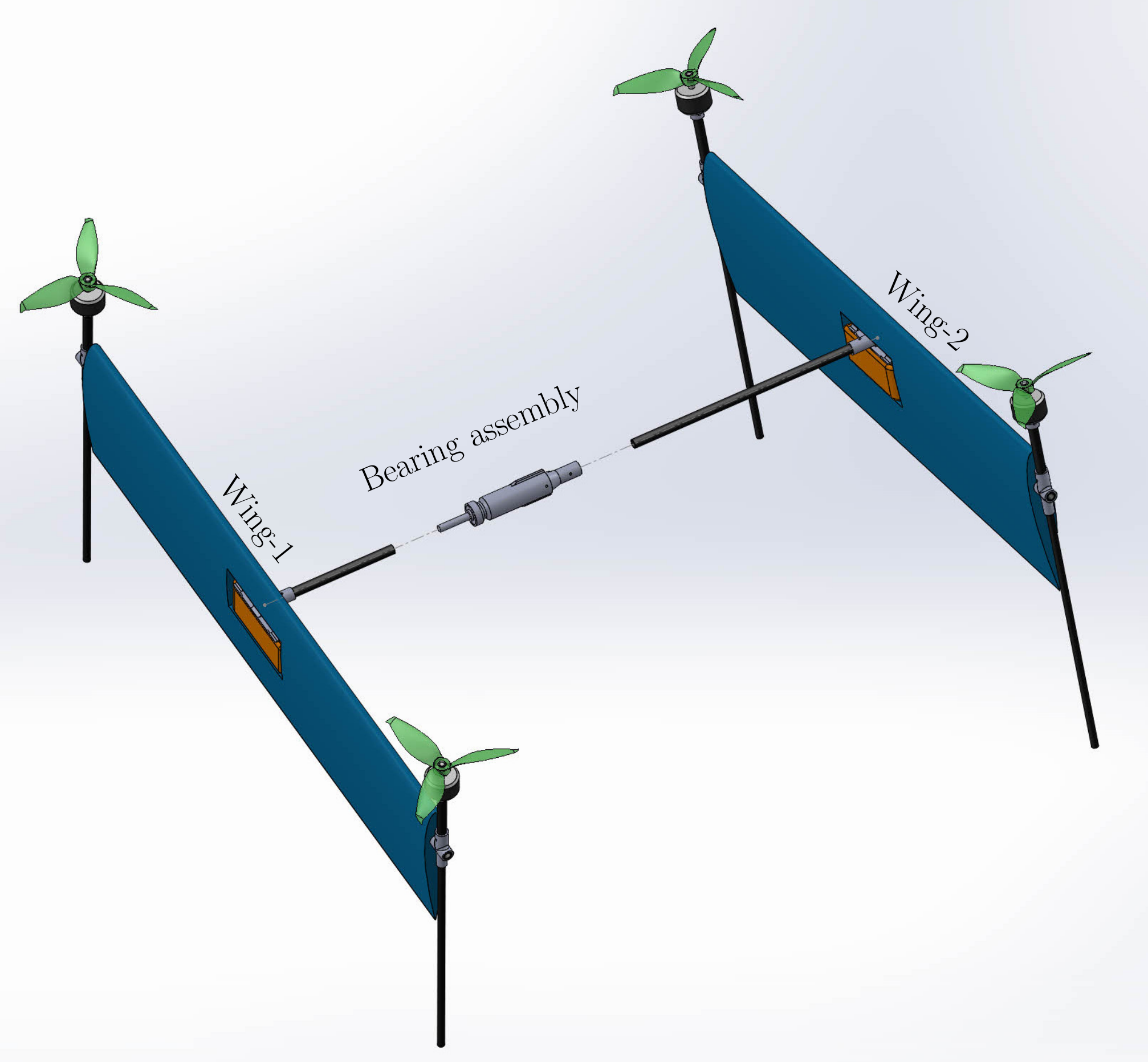}
\caption{Exploded view of swiveling biplane-quadrotor}
\label{fig:biplane_exploded}
\end{center}
\end{figure}

\begin{figure}[!htbp]
\begin{center}
\includegraphics[width=0.6\linewidth]{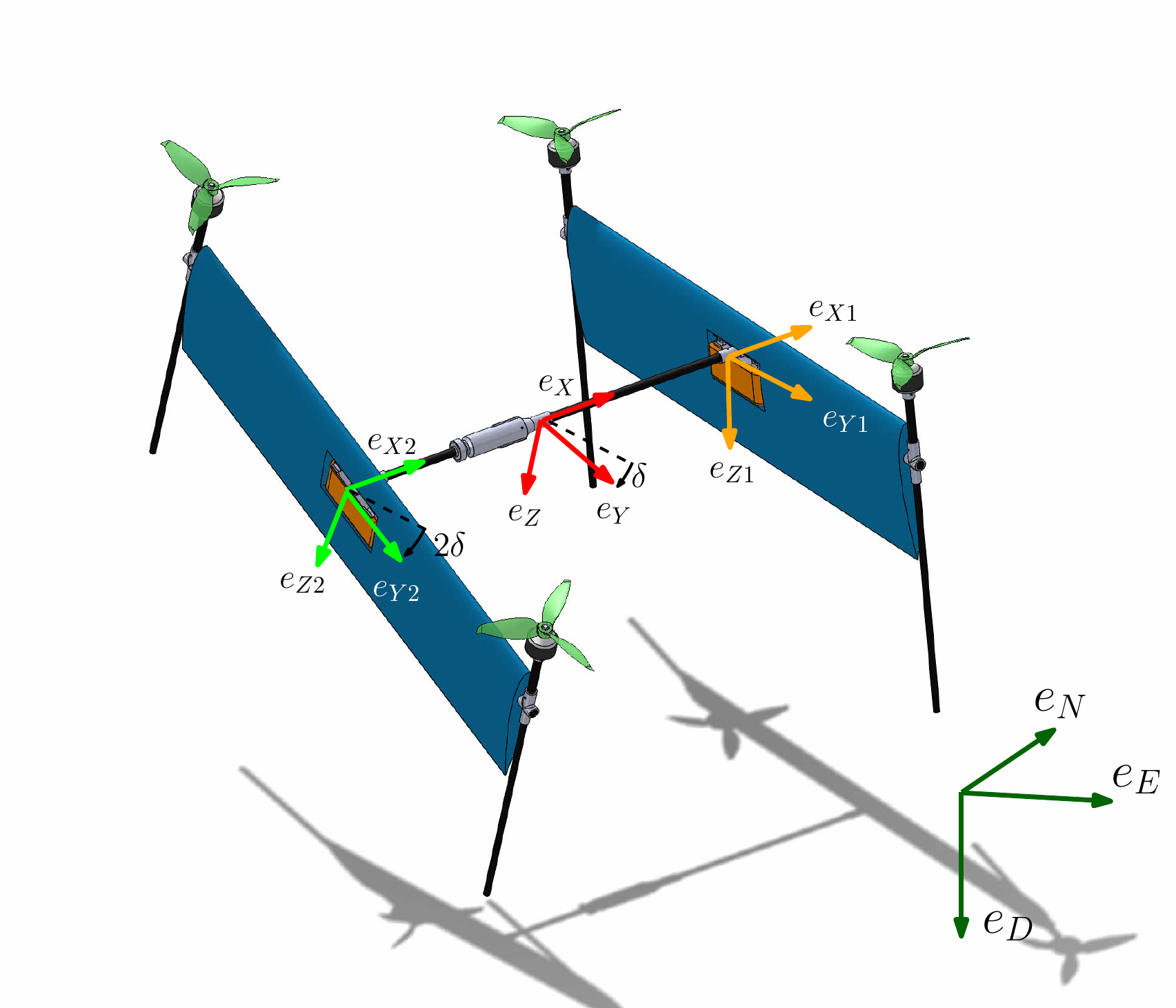}
\caption{Swiveling biplane-quadrotor with body and inertial frames}
\label{fig:biplane_model}
\end{center}
\end{figure}

The swiveling biplane quadrotor is made of two rigid wing-battery-motor  sub-substructures (henceforth called Wing-1 and Wing-2), each rigidly attached to a  slender rod as shown in Fig. \ref{fig:biplane_exploded}. The rods are attached to each other through a bearing assembly which allows relative rotational motion between the two wings about the rod axis. This results in a holonomic constraint and reduces the configuration space from $SO(3)\times SO(3)$ to $SO(3) \times S^1$. For developing the mathematical model of the vehicle, we approximate the rod to be massless and rigid enough to resist bending.

\subsection{Kinematics}
We define two orthogonal frames, Frame-1 $(e_{X1},e_{Y1},e_{Z1})$ and Frame-2 $(e_{X2},e_{Y2},e_{Z2})$ which form the principal inertia axes of Wing-1 and Wing-2 respectively. Defining swivel angle, $2\delta$, as the relative angle between the two frames about $e_{X1}$ axis. Let $R_1$ and $R_2$ denote respectively the rotation matrix from frames 1 and 2 to the inertial frame $(e_N,e_E,e_D)$. Let $R_{2\delta}$ denote the 1-dimensional rotation transformation from Frame-2 to Frame-1. The rotation   configuration of the vehicle could be entirely described by $(R_1,R_{2\delta})$.

The kinematics of the vehicle is given by,
\begin{equation}
\begin{split}
\dot{R}_1 = R_1 \hat{\omega}_1, \\
\dot{R}_{2\delta} = R_{2\delta} \hat{\omega}_{2\delta},
\end{split}
\end{equation}
where $\omega_1$ is the angular velocity of Frame-1 with respect to the inertial frame expressed in Frame-1, and $\omega_{2\delta} = [2\dot{\delta},0,0]$ is the swivel rate.  The hat map, $\hat{\cdot}$, is defined such that $\hat{a}b = a \times b $. Further, the angular velocity of Wing-2 may be expressed in Frame-2 as $\omega_2 = \omega_{2\delta} + R_{2 \delta}^T\omega_1$.

\begin{figure}[!htbp]
\begin{center}
\includegraphics[width=0.8\linewidth]{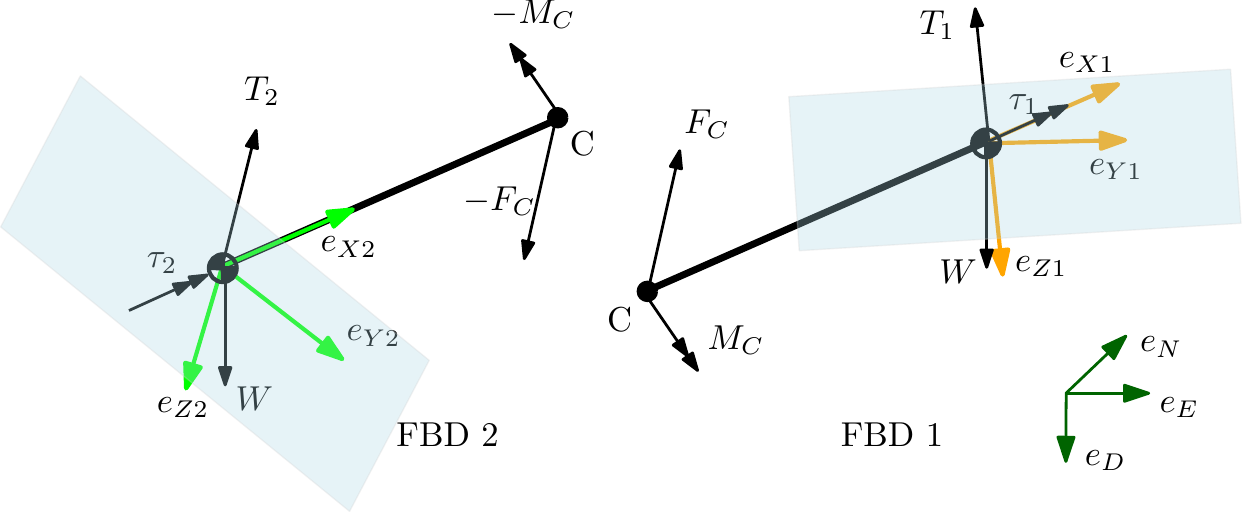}
\caption{Free body diagram of two wings separated at the combined center of mass C.}
\label{fig:fbd}
\end{center}
\end{figure}

\subsection{Dynamics}
Referring to the free body diagram 1 (FBD 1) shown in Fig. \ref{fig:fbd}, the angular momentum balance of Wing-1 about point C gives
\begin{equation}
\bm{\dot{H}_1} = \bm{J_1 \dot{\omega}_1} + \bm{\omega_1 \times J_1 \omega_1} = \bm{M_C} + \tau_1 \bm{e_{X1}} + l \bm{e_{X1}} \times (-T_1 \bm{e_{Z1}} + mg\bm{ e_D)}, 
\label{eq:AMB1}
\end{equation}
where $M_C$ is the reaction torque.
Similarly, from FBD 2, the angular momentum balance of Wing-2 about  point C gives
\begin{equation}
\bm{\dot{H}_2} = \bm{J_2 \dot{\omega}_2} + \bm{\omega_2 \times J_2 \omega_2} = -\bm{M_C} + \tau_2 \bm{e_{X2}} - l \bm{e_{X2}} \times (-T_2 \bm{e_{Z2}} + mg\bm{ e_D)}. 
\label{eq:AMB2} 
\end{equation}
The control input $T_i$ and $\tau_i$ represent the resultant thrust and torque about $e_{Xi}$-axis produced by the two motors on Wing-$i$, $i=1,2$. Note that, in the above equations, the terms in bold font are coordinate independent tensors and vectors. $\bm{J_1}$ and $\bm{J_2}$ are respectively the body-fixed moment of inertia tensors of Wing-1 and Wing-2 about point C. Using the constraint, $\bm{e_{X1}} = \bm{e_{X2}}$, and adding \eqref{eq:AMB1} and \eqref{eq:AMB2},
\begin{equation}
\bm{J_1 \dot{\omega}_1} + \bm{\omega_1 \times J_1 \omega_1} +
\bm{J_2 \dot{\omega}_2} + \bm{\omega_2 \times J_2 \omega_2} =
(\tau_1 + \tau_2)\bm{e_{X1}} + l (T_1\bm{e_{Y1}} - T_2\bm{e_{Y2}}).
\label{eq:dy1}
\end{equation}
Subtracting \eqref{eq:AMB1} from \eqref{eq:AMB2} and taking dot product with $\bm{e_{X1}}$, and using the free bearing condition $\bm{M_C} \cdot \bm{e_{X1}} = 0$, one obtains
\begin{equation}
(\bm{J_2 \dot{\omega}_2} + \bm{\omega_2 \times J_2 \omega_2} -
\bm{J_1 \dot{\omega}_1} - \bm{\omega_1 \times J_1 \omega_1}) \cdot \bm{e_{X1}} = \tau_2 - \tau_1.
\label{eq:dy2}
\end{equation}
Equations \eqref{eq:dy1} and \eqref{eq:dy2} constitute the coordinate system independent dynamics of the swiveling biplane quadrotor. 

For the purpose of designing an attitude controller, we derive the equation of motion in an intermediate frame, denoted by Frame-0 $(e_{X},e_{Y},e_{Z})$ (see Fig. \ref{fig:biplane_model}). Here, $e_X \triangleq e_{X1}$, $e_Y \triangleq (e_{Y1} + e_{Y2}) / \norm{(e_{Y1} + e_{Y2})}$, and $e_Z \triangleq (e_{Z1} + e_{Z2}) / \norm{(e_{Z1} + e_{Z2})}$. This frame suffers from singularity when $\delta = \pm 90$ degree as at this configuration $e_{Y1} + e_{Y2} = 0$ and $e_{Z1} + e_{Z2} = 0$. The singularity is far away from the normal operating range of vehicle with $\ \mid \delta \mid < \delta_{max} = 30 \deg $. Let $R$ denote the rotation matrix from Frame-0 to the inertial frame with the corresponding angular velocity $\omega$, then $\dot{R} = R \hat{\omega}$. The angular velocity of Frame-1 and Frame-2 can be expressed in terms of $\omega$ and $\omega_\delta = [\dot{\delta},0,0]$ as $\omega_1 = R_\delta \omega - \omega_\delta$ and $\omega_2 = R_\delta^T \omega + \omega_\delta$. The angular momentum of Wing-1 and Wing-2 about the point C expressed in Frame-0 are given by
\begin{equation}
\begin{split}
H_1 = R_\delta ^T J_1 \omega_1 = R_\delta ^T J_1 (R_\delta \omega - \omega_\delta) \\
H_2 = R_\delta  J_2 \omega_2 = R_\delta J_2 (R_\delta^T \omega + \omega_\delta)
\end{split}
\end{equation}
where $J_1$ and $J_2$ are respectively the representations of $\bm{J_1}$ and $\bm{J_2}$ in frames parallel to Frame-1 and Frame-2. Since Wing-1 and Wing-2 are identical in construction, $J_1 = J_2 = \text{diag}(J_{xx}, J_{yy}, J_{zz})$.
Then the total angular momentum about point C simplifies to 
\begin{equation}
H = H_1 + H_2 = \underbrace{(R_\delta ^T J_1 R_\delta + R_\delta J_2 R_\delta ^T)}_{J(\delta)} \omega = J(\delta) \omega
\end{equation}
with the equivalent inertia matrix
\begin{equation} \label{eq:inertia}
J(\delta) = 
 \begin{bmatrix}
    2J_{xx}       & 0 & 0 \\
    0       & 2\cos^2(\delta)J_{yy} + 2\sin^2(\delta)J_{zz} & 0 \\
    0       & 0 & 2\cos^2(\delta)J_{zz} + 2\sin^2(\delta)J_{yy}
\end{bmatrix}.
\end{equation}
To further utilize the symmetry in the system, the control input torques and thrusts are expressed in terms of their mean and differential components as
\begin{equation}
\begin{split}
T_1 = T_m + T_\Delta, \quad T_2 = T_m - T_\Delta \\
\tau_1 = \tau_m - \tau_\Delta, \quad \tau_2 = \tau_m + \tau_\Delta.
\end{split}
\end{equation} 
In terms of the new control inputs, \eqref{eq:dy1} is represented in Frame-0 from $\dot{H}_2 + \dot{H}_1$ as
\begin{equation} \label{eq:R_dyn}
J(\delta) \dot{\omega} + \dot{J}(\delta) \omega + \omega \times J(\delta) \omega = M
\end{equation}
where $M = [2\tau_m, 2lT_\Delta \cos(\delta), -2lT_m \sin(\delta) ]$. Similarly, \eqref{eq:dy2} is obtained in Frame-0 from $\dot{H}_2 - \dot{H}_1$ as
\begin{equation} \label{eq:del_dyn}
2J_{xx} \ddot{\delta} + \sin(2\delta)(J_{yy} - J_{zz})(\omega_y^2 - \omega_z^2) = 2\tau_\Delta
\end{equation}

\begin{remark}
From \eqref{eq:inertia}, it is worth noting that \textit{when $J_{yy} = J_{zz}$, the inertia tensor is invariant with respect to $\delta$}. Making the difference between $J_{yy}$ and $J_{zz}$ small enough could be set as a vehicle design requirement. For the vehicle configuration shown in Fig \ref{fig:biplane_model}, used for both simulation and experiments, $J_{zz} - J_{yy}$ is $0.4J_{zz}$.
\end{remark}

\section{Attitude Controller} \label{sec:att_control}

In this section, we present the attitude tracking problem and a controller for the swiveling biplane quadrotor. There are four degrees of freedom for actuation: one for total thrust, $T_m$, which is used for translational control and the remaining are three independent torques. There are four degrees of freedom for rotational configuration out of which only three are directly actuated ($\omega_z$ is underactuated). This makes the system underactuated in the rotational configuration. Nevertheless, it is possible to track an output whose dimension is equal to that of the input. We choose the attitude of the virtual frame $(e_X,e_Y,e_Z)$ (Fig. \ref{fig:biplane_model}) as the output to be tracked, since $e_Z$ represents the total  thrust direction, which would make the translational control of the vehicle analogous to that of a regular quadrotor. In subsection \ref{subsec:att_des} a method of specifying the desired attitude trajectory which facilitates easy manual transition is presented. The next subsection presents the attitude tracking controller and stability results. 

\subsection{Reference Attitude Trajectory} \label{subsec:att_des}
\begin{figure}[!htbp]
\begin{center}
\includegraphics[width=0.9\linewidth]{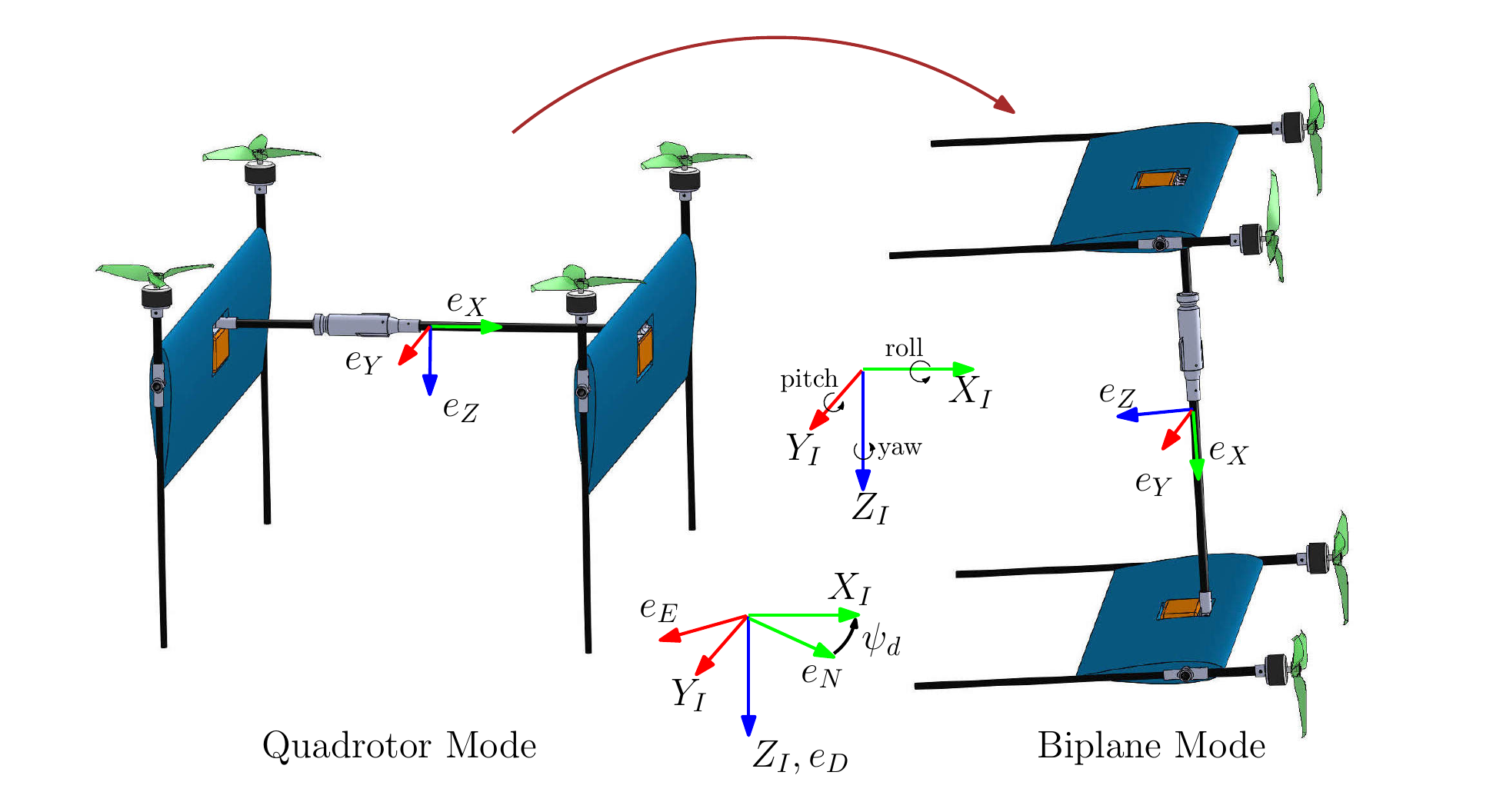}
\caption{Swapping of roll and yaw axis as the vehicle transitions. In the quadrotor mode, roll command should rotate the vehicle about body fixed $e_X$ axis, whereas in the biplane mode the same command should rotate it about body fixed $e_Z$ axis. A similar swapping happens for the yaw command. This makes manual piloting difficult during the transition phase, if one chooses to use just a rate stabilized controller.}
\label{fig:mode_change}
\end{center}
\end{figure}
For a single attitude controller to be valid uniformly for both the quadrotor mode and biplane forward flight mode and throughout the transition it is necessary to have a singularity free expression for the desired attitude. In this regard, the desired rotation matrix is expressed as a function of the 312 Euler angle sequence with the following three advantages facilitating easy manual flight.
\begin{enumerate}
\item In the manual flight mode, the desired attitude or angular rate is specified using stick input from the pilot. If the pilot flies the vehicle in angular rate stabilized mode, the stick input would correspond to body frame angular rates. In the quadrotor mode a roll-rate command would specify angular rate about $e_X$ axis as shown in Fig. \ref{fig:mode_change}, which would roll the vehicle about $e_X$ axis as expected. However, in the biplane mode the same roll rate command would rotate the vehicle about $e_X$ axis which corresponds to a yaw motion. Instead,  if the pilot stick input corresponds to 312 Euler angle, the roll and yaw commands are rotations about inertial frame $Z_I$ axis and a horizontal axis $X_I$ as shown in Fig. \ref{fig:mode_change}. This results in an intuitive correspondence of the pilot stick input with physical rotation of the vehicle in both the flight modes and throughout the transition.
\item The conventionally used 321 Euler angle has a singularity (gimbal lock) at pitch angle of $\pm 90$ degree, which is an operating point of the vehicle. Although the 312 sequence has singularity at roll angles of $\pm 90$ degree, we restrict the commanded roll angle to operate well within this range.
\item Another advantage of 312 sequence is that for a given pitch angle, a change in roll command rotates the body frame X axis around a cone, thus maintaining the angle of attack during a level flight as shown in Fig. \ref{fig:Euler312}. This facilitates an easy manual control of the vehicle during transition and biplane mode unlike the 321 Euler sequence wherein a roll command rotates the vehicle about body fixed $e_X$-axis, which in turn introduces side slip.
\end{enumerate}
\begin{figure}[!htbp]
\begin{center}
\includegraphics[width=0.5\linewidth]{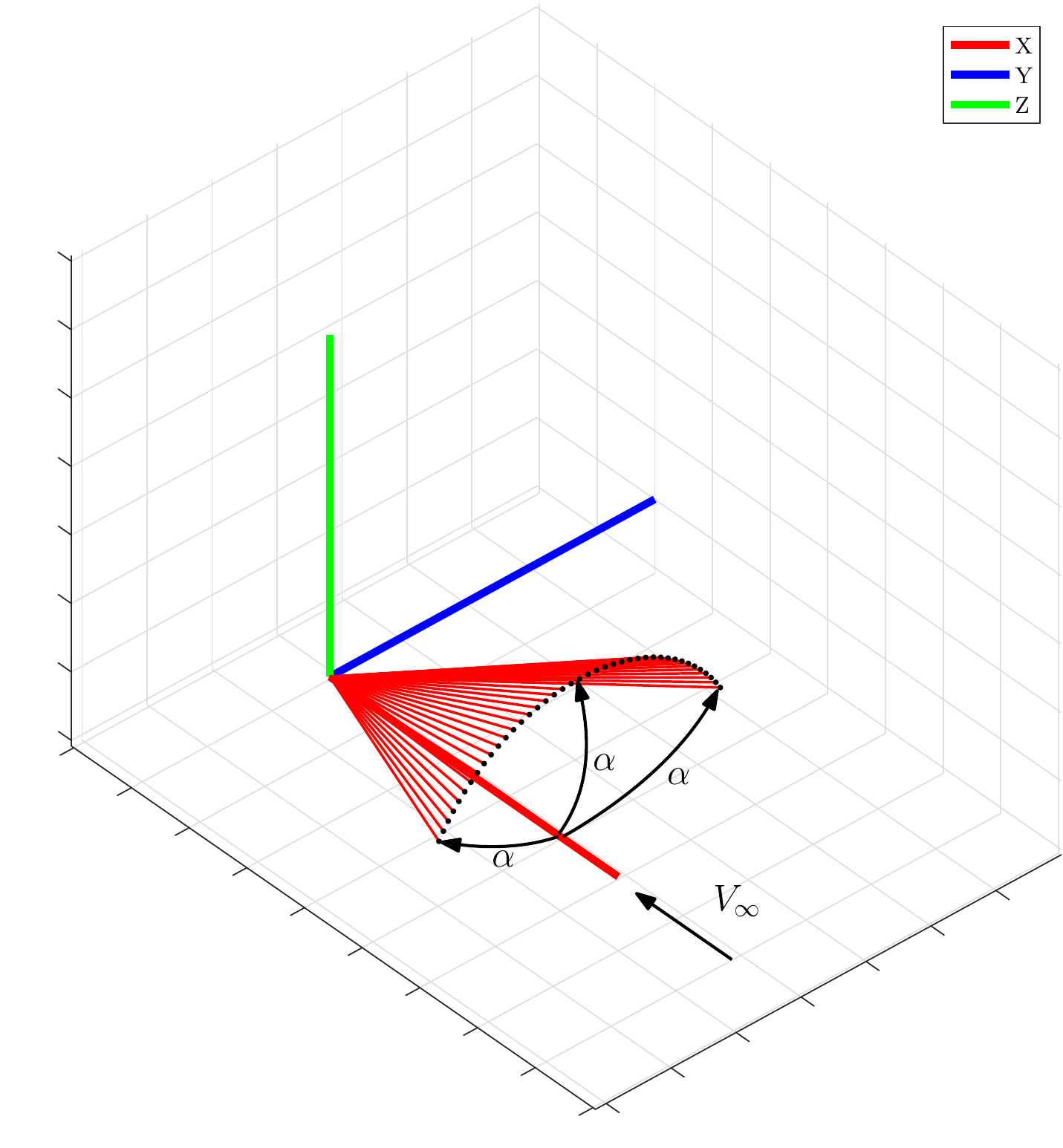}
\caption{Body frame X-axis, $e_X$, traces out a cone maintaining constant AoA, as 312 Euler roll angle is varied from -60 to 60 deg while pitch is kept at 30 deg. A roll variation in 312 Euler sequence rotates the body frame about a horizontal axis, whereas in 321 Euler sequence the body frame rotates about body frame X-axis.}
\label{fig:Euler312}
\end{center}
\end{figure}

\noindent The desired rotation matrix $R_d$ in terms of the desired 312 Euler angle $(\psi_d,\phi_d,\theta_d)$ is given by
\begin{equation}
R_d =  \begin{bmatrix}
C_{\theta_d} C_{\psi_d} -S_{\phi_d} S_{\theta_d} S_{\psi_d} 
&-C_{\phi_d} S_{\psi_d} 
&S_{\theta_d} C_{\psi_d} +S_{\phi_d} C_{\theta_d} S_{\psi_d} \\ 
C_{\theta_d} S_{\psi_d} +S_{\phi_d} S_{\theta_d} C_{\psi_d} 
&C_{\phi_d} C_{\psi_d} 
&S_{\theta_d} S_{\psi_d} -S_{\phi_d} C_{\theta_d} C_{\psi_d} \\ 
-C_{\phi_d} S_{\theta_d} 
&S_{\phi_d} 
&C_{\phi_d} C_{\theta_d}
\end{bmatrix} 
\end{equation}
where $C_{\ast}$, $S_{\ast}$ are short for $\cos(\ast)$ and $\sin(\ast)$.
Here, the desired pitch angle $\theta_d = \theta_c + \theta_{tr}$, where $\theta_c$ is the stick input from the pilot and $\theta_{tr}$ is trim input which could be assigned to a separate knob to continuously vary the pitch from $0$ to $90 \deg$ during the transition.  

\subsection{Attitude Tracking Controller} \label{subsec:att_con}
The objective of the control design is for the virtual frame attitude $R$ to track a sufficiently smooth reference command $R_d(t)$. The model, given by \eqref{eq:R_dyn} and \eqref{eq:del_dyn}, does not allow for conventional output tracking control techniques, such as dynamics feedback linearization \cite{isidori2013nonlinear} to be applied due to a singularity at the nominal operating point of $\delta = 0$ configuration. The singularity arises due the the configuration dependent inertia, $J(\delta)$, and could be observed in an attempt to perform dynamic feedback linearization as follows. Consider the last component of \eqref{eq:R_dyn}, corresponding to $M_z$
\begin{equation} \label{eq:Mzdd}
2(\cos^2(\delta) J_{zz} + \sin^2(\delta) J_{yy})\dot{\omega}_z + 2\dot{\delta} \sin{2\delta}(J_{yy}-J_{zz})\omega_z + (2\cos^2(\delta) J_{yy} + 2 \sin^2(\delta) J_{zz} - J_{xx})\omega_y \omega_x = M_z.
\end{equation}
Since $M_z$ is not a control input, following the procedure for feedback linearization, differentiate the above equation until an input explicitly appears. It is observed that differentiating it once results in the control input $\tau_\Delta$ appearing through $\ddot{\delta}$ in the second term of \eqref{eq:Mzdd} as $2\ddot{\delta} \sin(2\delta)(J_{yy} -J_{zz})\omega_z$. The coefficient $\sin(2\delta)$ of $\tau_\Delta$ vanishing at $\delta = 0$ resulting in the singularity.

Therefore, to circumvent this for the purpose of controller design, we use a nominal model with constant inertia tensor associated with the $\delta = 0$ configuration. This reduces \eqref{eq:R_dyn} to 
\begin{equation}\label{eq:R_dyn_control}
J \dot{\omega} + \omega \times J \omega = M,
\end{equation}
where $J = \text{diag}(2J_{xx},2J_{yy},2J_{zz})$.
This model is valid for the entire operational range of $\delta$ as it is limited to $\pm 30 \deg$ and any error due to this approximation is expected to be handled by the  feedback controller. 

The $z$ component of $M$ in \eqref{eq:R_dyn_control}, $M_z$, is not a control parameter and is dependent of $\delta$ which has a second order dynamics. The total thrust $T_m$ is modified as $\frac{T_0}{\cos(\delta)}$ to obtain the following diffeomorphism between $\delta \in (-\pi/2, \pi/2)$ and  $M_z \in \mathbb{R}$ as
\begin{equation}
M_z = -2lT_0 \tan(\delta),
\end{equation}
with their derivatives,
\begin{equation} \label{eq:M_z_rel}
\begin{aligned}
\dot{M}_z &= -2lT_0 \sec ^2(\delta) \dot{\delta}, \\
\ddot{M}_z &= -4lT_0 \sec ^2(\delta) \tan(\delta) \dot{\delta}^2 -2lT_0 \sec ^2(\delta) \ddot{\delta}.
\end{aligned}
\end{equation}
This modified thrust also ensures that the component of thrust along $\bold{e_Z}$ does not change with $\delta$. Feedback linearization of $\delta$ dynamics \eqref{eq:del_dyn} would result in $\ddot{\delta} = v_z$ with 
\begin{equation}
v_z =  \frac{\tau_\Delta}{J_{xx}} - \frac{J_{yy} - J_{zz}}{2J_{xx}}\sin(2\delta)(\omega_y^2 - \omega_z^2).
\end{equation}
Using the above in \eqref{eq:M_z_rel} would result in $\ddot{M}_z = u_z$ with the new control input
\begin{equation}
u_z = -4lT_0 \sec ^2(\delta) \tan(\delta) \dot{\delta}^2 -2lT_0 \sec ^2(\delta) v_z
\end{equation}   
In order to achieve a vector relative degree for this MIMO system, the dynamics of $M_x$ and $M_y$ are extended as a double integrator. For a detailed exposition of dynamic feedback linearization and the concept of vector relative degree the reader is referred to Isidori (Sec 5.4 of \cite{isidori2013nonlinear}). This results in the following nominal model 
\begin{subequations} \label{eq:control_dynamics}
\begin{equation} \label{eq:rigid_body_a}
\dot{R} = R \hat{\omega},
\end{equation}
\begin{equation} \label{eq:rigid_body_b}
J \dot{\omega} + \omega \times J \omega = M,
\end{equation}
\begin{equation}  \label{eq:moment_dyn}
\ddot{M} = u,
\end{equation}
\end{subequations}
where $M = [M_x,M_y,M_z]^T$ and $u = [u_x,u_y,u_z]^T$ is the new control input. Since the sum of vector relative degree about each axes (4+4+4) is equal to the dimension of the state space of the extended system, the feedback linearization has trivial zero dynamics \cite{isidori2013nonlinear}. The authors' previous work on attitude tracking of aerobatic helicopters \cite{raj2018robust} has resulted in rigid body dynamics augmented with first order moment dynamics.

Equations \eqref{eq:rigid_body_a} and \eqref{eq:rigid_body_b} represent the rigid body dynamics and the solution to its tracking problem has been obtained previously with geometric control techniques \citep{maithripala2006almost, koditschek1989application, chaturvedi2009asymptotic, bayadi2014almost, lee2011robust_adaptive}. Since the equations are in strict feedback form, one could use backstepping technique to design a tracking controller. Although it is easier to show stability results with backstepping technique, there are two disadvantages with it: 1) this would lead to introduction of certain control terms which are solely meant to make the derivative of Lyapunov function negative definite but otherwise not useful, and 2) gain selection could be relatively difficult. Instead, we adopt a more intuitive, two stage design process owing to the special structure of the second order moment dynamics. First, an attitude tracking controller is designed for the rigid body using existing geometric control technique to obtain the desired control moment $M_d$. Next, M is made to track $M_d$ by assigning it the familiar spring-mass-damper structure using the control input $u$. 

First, to design the controller for rigid body, the result by Maithripala  \citep{maithripala2006almost} is used which states that any tracking problem of a fully actuated mechanical system on a Lie group could be converted to a stabilization problem. The stabilization problem for a mechanical system defined on a manifold is solved with a proportional-derivative (PD) control structure \citep{koditschek1989application}. The proportional action is derived from the gradient of a \textit{configuration error function} $\psi : SO(3) \to \mathbb{R}$ which satisfies the following lemma.
\begin{lemma} [Chillingworth \cite{chillingworth1982symmetry}]
If $P \in \mathbb{R}^{3 \times 3}$ is symmetric positive definite with distinct eigenvalues, then the configuration error function
\begin{equation}
\psi \triangleq \frac{1}{2} tr(P(I-R_e))
\end{equation}
has exactly 4 critical points.
\end{lemma}
\noindent The critical points of $\psi_m$ are $\Theta = \{I, e^{\pi \hat{v}_1}, e^{\pi \hat{v}_2}, e^{\pi \hat{v}_3} \}$, where $v_1, v_2, v_3$ are the eigenvectors of $P$ (p. 553 \citep{bullo2004geometric}). The gradient of $\psi$ when pulled back to the Lie algebra is given by
\begin{equation}
e_R = \frac{1}{2}(PR_e - R_e^T P),
\end{equation}
and $e_R = 0$ when $R_e \in \Theta$.

Given a smooth reference attitude trajectory $R_d(t)$ which satisfies $\dot{R}_d = R_d \omega_d$, the rotation error is defined as $R_e \triangleq R_d^T R$ and satisfies $\dot{R}_e = R_e \hat{e}_\omega$, where $e_\omega = \omega - R_e^T \omega_d$. Then, the tracking error dynamics in terms of the error configuration is given by 
\begin{equation} \label{eq:rigid_body_err_dyn}
\begin{gathered}
\dot{R}_e = R_e \hat{e}_\omega \\
J \dot{e}_\omega = -\omega \times J\omega + J(\hat{e}_\omega R_e^T \omega_d - R_e^T \dot{\omega}_d) + M
\end{gathered}
\end{equation}
The following lemma gives the attitude tracking controller for rigid body dynamics.
\begin{lemma} [Rigid body attitude tracking] 
For $k_R>0$, $k_\omega > 0$ the control moment
\begin{equation} \label{eq:rigid_body_control}
M = -k_R e_R - k_\omega e_\omega + \omega \times J\omega -  J(\hat{e}_\omega R_e^T \omega_d - R_e^T \dot{\omega}_d)
\end{equation}
renders the desired equilibrium $(I,0)$ of the error dynamics \eqref{eq:rigid_body_err_dyn} almost-globally asymptotically stable.
\end{lemma} 
\noindent The proof follows form the one given for stabilizing controller in \citep{chaturvedi2009asymptotic}.

Due to the second order dynamics of moment in \eqref{eq:control_dynamics}, we define the moment error $M_e \triangleq M - M_d$ where $M_d$ is the desired moment given by r.h.s of \eqref{eq:rigid_body_control}. To make the moment error dynamics stable, we assign it the familiar spring-mass-damper structure by choosing control input as
\begin{equation} \label{eq:control}
u = \ddot{M}_d - D\dot{M}_e - KM_e,
\end{equation} 
where the damping matrix $D = \text{diag}(2\zeta_x\Omega_{x},2\zeta_y\Omega_{y},2\zeta_z\Omega_{z})$ and stiffness matrix $K = \text{diag}(\Omega_{x}^2,\Omega_{y}^2, \Omega_{z}^2)$. Here $\zeta_i$ and $\Omega_{i}$ ($i \in \{x,y,z\}$) are respectively the damping coefficient and natural frequency of the individual axes. With the control input \eqref{eq:control}, the error dynamics of the dynamically extended nominal system \eqref{eq:control_dynamics} reduces to 
\begin{subequations}  \label{eq:err_dyn}
\begin{equation}
\dot{R}_e = R_e \hat{e}_\omega
\end{equation} 
\begin{equation}
J \dot{e}_\omega = -k_R e_R - k_\omega e_\omega + M_e 
\end{equation}
\begin{equation} \label{eq:err_dyn_Me}
\ddot{M}_e + D\dot{M}_e + KM_e = 0.
\end{equation}
\end{subequations}

\noindent The above dynamical system has 4 equilibrium points characterized by $(e_R,e_\omega, M_e, \dot{M}_e) = (0,0,0,0)$ and explicitly $(R_{eq}\in \Theta,0,0,0)$. The following theorem states the stability properties of the error dynamics \eqref{eq:err_dyn} of the nominal system \eqref{eq:control_dynamics}.

\begin{theorem}
For $k_R>0$, $k_\omega > 0$, $\zeta_i > 0$ and $\Omega_{i}> 0$, the control input $u$ given by \eqref{eq:control}, renders the desired equilibrium $(I,0,0,0)$ of the error dynamics \eqref{eq:err_dyn} almost-globally asymptotically stable.
\end{theorem}

\begin{proof}
We prove the almost-global asymptotic stability in two steps. First, we linearize the error dynamics about each of the four equilibrium points to check the dimension of stable manifold of each one. The linearization is used to show that the stable manifold associated with the undesired equilibrium points is a thin set. Next, LaSalle's invariance principle is used to show that all the initial conditions in the state space converge to one of the four equilibrium points. 

To linearize the error dynamics, differentiate the perturbation of equilibrium state $(R_{eq}e^{\epsilon \hat{\bar{\eta}}}, \epsilon \bar{e}_\omega  , \epsilon \bar{M}_e , \epsilon \dot{\bar{M}}_e)$ with respect to $\epsilon$ evaluated at $\epsilon = 0$, with the linearized system state being $(\bar{\eta}, \bar{e}_{\omega}, \bar{M}_e, \dot{\bar{M}}_e) \in \mathbb{R}^{12}$. The linearized equation is given by
\begin{equation} \label{eq:linear_dyn}
\frac{d}{dt} \begin{bmatrix}
\bar{\eta} \\
\bar{e}_\omega \\
\bar{M}_e \\
\dot{\bar{M}}_e
\end{bmatrix} = 
\underbrace{
\begin{bmatrix}
\pmb{0} & \pmb{I} & \pmb{0}  & \pmb{0}\\
-J^{-1} k_R B(R_{eq}) & -J^{-1}k_\omega & J^{-1} & \pmb{0}  \\
\pmb{0} & \pmb{0} & \pmb{0} & \pmb{I}  \\
\pmb{0} & \pmb{0} & -K & -D  
\end{bmatrix}}_{\textstyle S(R_{eq})}
\begin{bmatrix}
\bar{\eta} \\
\bar{e}_\omega \\
\bar{M}_e \\
\dot{\bar{M}}_e
\end{bmatrix}
\end{equation}
where $B(R_{eq}) = -\frac{1}{2}\sum_{i=1}^3 \hat{e}_i P R_{eq} \hat{e}_i$, and $\pmb{I}$ is the $3\times 3$ identity matrix. Details of the linearization procedure are given in \citep{raj2019heli}. For positive values of the gains $k_R, k_\omega, \zeta_i, \Omega_{i} $ it is observed that $S(R_{eq})$ is hyperbolic for all $R_{eq} \in \Theta$. The real part of eigenvalues of $S(R_{eq})$ reveal that the desired equilibrium point $(I,0,0,0)$ is stable and the rest have at least one eigenvalue that has positive real part. Therefore it follows that the stable eigenspace of unstable equilibria are of dimension less than the state space dimension. This in turn implies that the corresponding stable manifold associated with the unstable equilibria are of dimension less than state space and hence of is a thin set.

For applying LaSalle's invariance principle, the following candidate Lyapunov function is chosen which is the sum of Lyapunov function for the rigid body and the spring-mass system
\begin{equation}
V = k_R \psi(R_e) + \frac{1}{2}e_\omega^T J e_\omega + \frac{1}{2} M_e^T KM_e + \frac{1}{2} \dot{M}_e^T \dot{M}_e.  
\end{equation}
The directional derivative of the above function along the vector field of the error dynamics \eqref{eq:err_dyn} is
\begin{equation}
\dot{V} = -e_\omega^T (k_\omega e_\omega  -  M_e) - \dot{M}_e^T D \dot{M}_e 
\end{equation}
Since the moment error dynamics \eqref{eq:err_dyn_Me} is linear and stable, $\norm{M_e}$ decreases exponentially and $\norm{M_e(t)} < M_0 e^{-\lambda t}$ for some $M_0>0$ determined by the initial conditions $(M(0),\dot{M}(0))$, and $\lambda > 0$ determined from the gains $\zeta_i$ and $\Omega_{i}$. Therefore using the upper bound on $\norm{M_e}$ and Cauchy-Schwarz inequality, $\dot{V}$ is upper bounded by
\begin{equation} \label{eq:Lyp_ineq}
\dot{V} \leq -\norm{e_\omega} (\underbrace{k_\omega \norm{e_\omega} - M_0 e^{-\lambda t}}_{c(t)})  - \dot{M}_e^T D \dot{M}_e.
\end{equation}
If the gains are chosen such that: 1) $\norm{e_\omega}$ decays slower than $e^{-\lambda t}$, and 2) $e_\omega$ does not have oscillatory behavior (overdamped), then for all initial conditions, there is a finite time $t_0>0$ such that $c(t)>0$ for all $t>t_0$. An explanation for such a choice of gains is given in Sec. \ref{sec:gain}. Therefore, for $t>t_0$, $\dot{V}$ is negative semi-definite and LaSalle's invariance argument is valid. Since $V$ is bounded from below, all initial conditions converge to the largest invariant set within the set characterized by $\dot{V} \equiv 0$. Hence, from \eqref{eq:Lyp_ineq} the following arguments follow
\begin{equation}
\dot{V} \equiv 0 \implies e_\omega \equiv 0 \quad \text{and} \quad \dot{M}_e \equiv 0. 
\end{equation}
From \eqref{eq:err_dyn} and above
\begin{equation}
\dot{M}_e \equiv 0 \implies M_e \equiv 0, \quad e_\omega \equiv 0 \implies e_R \equiv 0.
\end{equation}
Therefore the limit set for all initial conditions in the state space is characterized by $(e_R,e_\omega, M_e, \dot{M}_e) = (0,0,0,0)$ which is the set of equilibrium points of the error dynamics \eqref{eq:err_dyn}. Hence, all initial conditions, except for the stable manifold associated with the unstable equilibria, converges to the desired equilibrium $(I,0,0,0)$.

\end{proof}

\begin{remark}
The proposed controller exhibits almost-global stability in the $M_z$ space which corresponds to $\delta$ belonging to $(-\pi/2,\pi/2)$. If the initial condition is such that $\delta(0) \in (-\pi/2,\pi/2)$, then the Frame-0 attitude  converges to the desired attitude for almost all initial conditions. This is a reasonable assumption as the vehicle always starts close to $\delta = 0$ configuration and the operating range of $\delta$ is well within $(-\pi/2,\pi/2)$.
\end{remark}

\begin{figure}[!htbp]
\begin{center}
\includegraphics[width=0.9\linewidth]{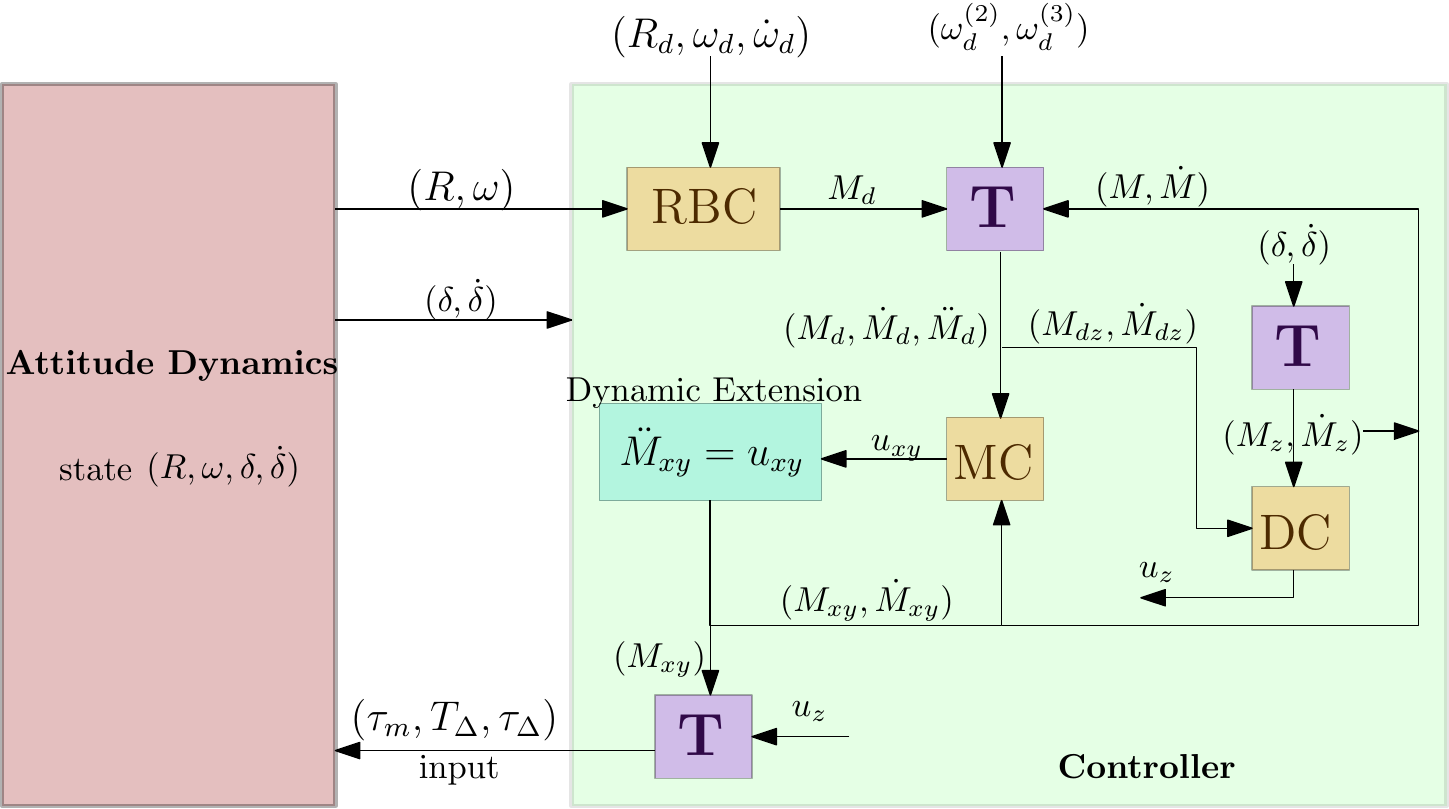}
\caption{Block diagram showing the structure of the attitude controller. Blocks labeled \textbf{T} are appropriate algebraic transformations. Blocks labeled RBC, MC and DC respectively represent rigid-body attitude controller, moment ($M_x, M_y$) controller and $\delta$ ($M_z$) controller.}
\label{fig:control_block}
\end{center}
\end{figure}

\subsection{Controller structure}
The nominal dynamics \eqref{eq:control_dynamics} resembles the structure of two double integrators in cascade. 
\begin{equation}
\begin{gathered}
\ddot{X} = Y \\
\ddot{Y} = v
\end{gathered}
\end{equation}
Here, the tracking output $X \in \mathbb{R}^3$ with the desired command $X_d$. By assigning the individual double integrators a spring-mass-damper structure, the following error dynamics is obtained which resembles the linearized error dynamics \eqref{eq:linear_dyn}
\begin{equation} \label{eq:linear_dyn1}
\frac{d}{dt} \begin{bmatrix}
X_e \\
\dot{X}_e \\
Y_e \\
\dot{Y}_e
\end{bmatrix} = 
\begin{bmatrix}
\pmb{0} & \pmb{I} & \pmb{0}  & \pmb{0}\\
-K_X & -D_X & \pmb{I} & \pmb{0}  \\
\pmb{0} & \pmb{0} & \pmb{0} & \pmb{I}  \\
\pmb{0} & \pmb{0} & -K_Y & -D_Y  
\end{bmatrix}
\begin{bmatrix}
X_e \\
\dot{X}_e \\
Y_e \\
\dot{Y}_e
\end{bmatrix},
\end{equation}
where $X_e = X - X_d$, $Y_e = Y - Y_d$, $Y_d = \ddot{X}_d - D_X\dot{X}_e - K_X X_e$ and control input $v = \ddot{Y}_d - D_Y\dot{Y}_e - K_Y Y_e$. The gain matrices $K_X$, $K_Y$ and $D_X$, $D_Y$ are composed of natural frequency and damping parameters as in \eqref{eq:control}.  The error dynamics has the property that the eigenvalues corresponding to the inner loop and outer loop are independent and equal to their corresponding natural frequencies assigned. This makes the process of choosing individual rate of convergence for the inner and outer loops of error dynamics.

\subsection{Gain selection} \label{sec:gain}
We refer the second order moment dynamics \eqref{eq:moment_dyn} as the inner loop and the rigid body dynamics \eqref{eq:rigid_body_a}, \eqref{eq:rigid_body_b} as the outer loop. In the inner loop, $M_z$ dynamics corresponds to $\delta$ dynamics and is controlled by physical actuators ($\tau_\Delta$). Whereas, the $M_x$ and $M_y$ dynamics are dynamic extension and is part of the controller and is not affected by actuator bandwidth. Similarly, $\omega_x$ and $\omega_y$ dynamics is also actuated directly by the motors. The gains of the individual loops directly actuated by physical actuators ($\omega_x, \omega_y, \dot{\delta}$) are limited by the actuator bandwidth. 
The upper 2x2 block diagonal matrix of \eqref{eq:linear_dyn} represent the attitude dynamics and the lower 2x2 block diagonal represent the moment dynamics. Here, $K$ and $J^{-1}k_RB(R_{eq})$ are stiffness terms representing the equivalent natural frequencies, and $D$ and $J^{-1}k_\omega$ are the damping terms representing the equivalent damping ratio. Oscillatory behavior in the response of $e_\omega$ could be avoided, as required in the proof of Theorem 1, by assigning a damping ratio greater than one. Similarly, by assigning a higher stiffness to the inner loop (moment dynamics) than the outer loop, one could ensure a faster decay of $M_e$ compared to $e_\omega$. 

\begin{figure}[!htbp]
\begin{center}
\includegraphics[width=1.0\linewidth]{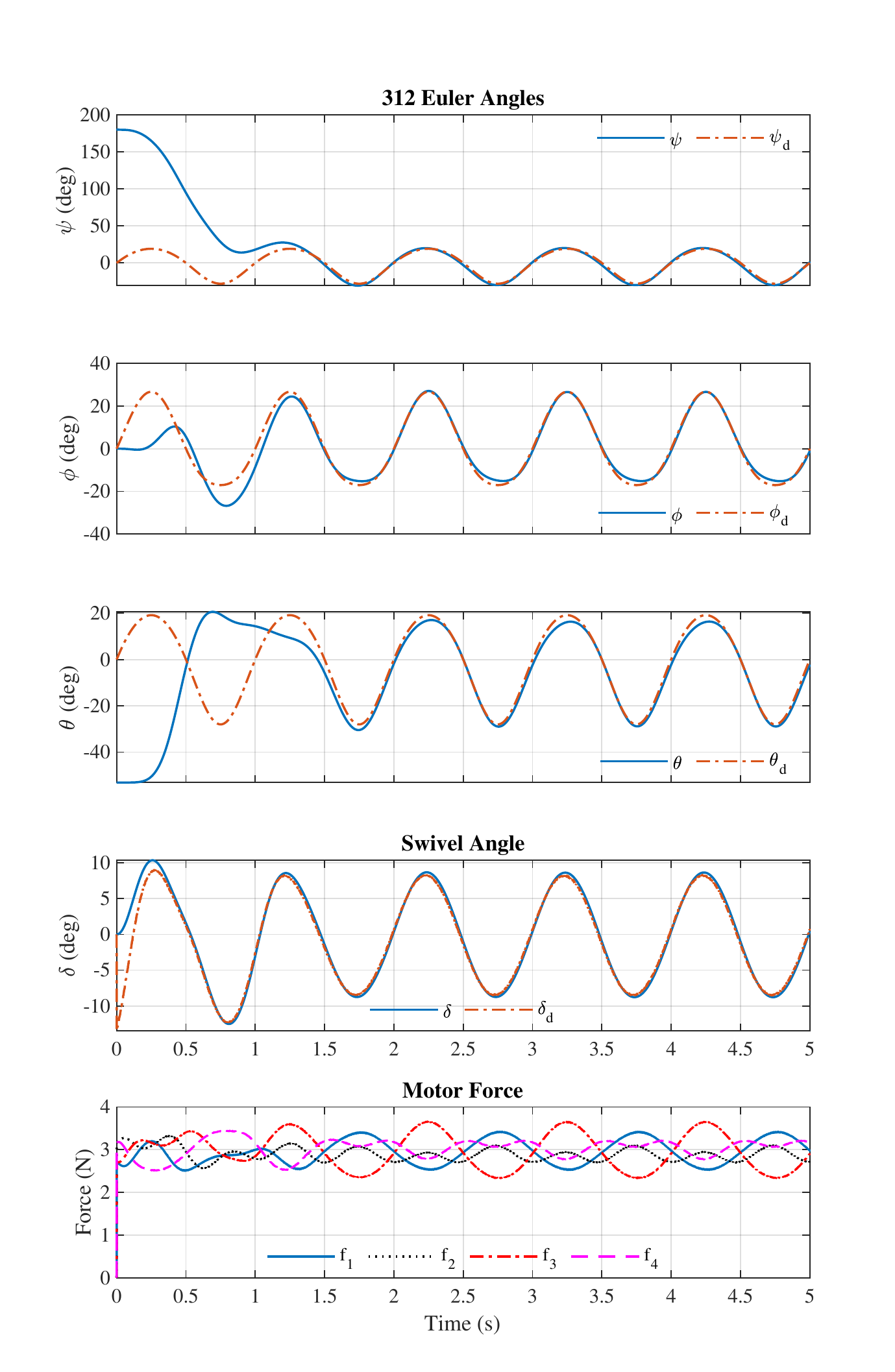}
\caption{Performance of the proposed controller in simulation with 5 percent inertia uncertainty, noisy measurement, and actuator dynamics.}
\label{fig:sim}
\end{center}
\end{figure}

\section{Numerical Simulation} \label{sec:sim}
Numerical simulations were carried out to check the efficacy of the controller in the presence of unmodeled dynamics as the controller is based on a nominal model with constant inertia, has dynamic extension component and involves cancellation of certain nonlinearities. Simulation was carried out on a model whose parameters (see Table \ref{tab:biplane_params}) were derived from the vehicle used for experimentation. 
The following three disturbances were incorporated in the simulation model.
\begin{enumerate}
\item Inertia was varied upto 5 percent of the true value in the simulation. The true value of inertia had to be obtained from CAD model.
\item A first order motor dynamics with a time constant of 0.015 s was used to approximate the actuator dynamics. The time constant was obtained from load cell measurements.
\item Additive Gaussian white noise with zero mean and 0.075 rad/s standard deviation models effect of vibration on angular velocity measurement. The value of the standard deviation was arrived based on the flight data recorded in hover condition. The configuration states $(R,\delta)$ are not affected by high frequency noise as they are obtained from estimators after integration.
\end{enumerate}

As a reference output trajectory, the virtual frame was made to track a sinusoidal of 40 deg amplitude and 1 Hertz frequency about spatial frame fixed (1,1,1) axis. This reference signal will sufficiently excite the underactuated $\omega_z$ dynamics and the gyroscopic moments. A large initial attitude error (180,0,50) deg in terms of 312 Euler angles was used. The bock diagram structure of the controller is shown in Fig. \ref{fig:control_block}. 
The performance of the controller based on the nominal vehicle dynamics with the above mentioned disturbance is shown in Fig. \ref{fig:sim}. The configuration states converge to its reference value in under 2 seconds and the actuator command required to do so is well within the motor force limits of 6.74 N. It is evident that the added uncertainties and nonlinearity cancellation have not introduced instability in the closed loop dynamics and the tracking performance is reasonable.

\begin{table} [tbp!]
\renewcommand{\arraystretch}{1.3}
\caption{Swiveling biplane parameters}
\label{tab:biplane_params}
\centering
\begin{tabular}{l l l} 
 \hline
 Parameter & Description & Values  \\
 \hline
 $[J_{xx}, J_{yy}, J_{zz}]$ & Moment of inertia & [1.111, 1.36, 2.275]$\times10^{-2}$ $kg$-$m^2$  \\ 

 $l$  & Wing separation & 42 $cm$ \\

 $L$  & Motor separation on Wing-1,2 & 61 $cm$ \\
 
 $m$ & Vehicle mass & 800 $g$\\ 
 
 $\tau_m$  & Motor time constant & 0.015 $s$  \\  
 
 $F_{max}$ & Maximum motor force & 6.74 $N$ \\ [1ex]
 \hline
\end{tabular}
\end{table} 

\begin{figure}[!htbp]
\begin{center}
\includegraphics[width=0.6\linewidth]{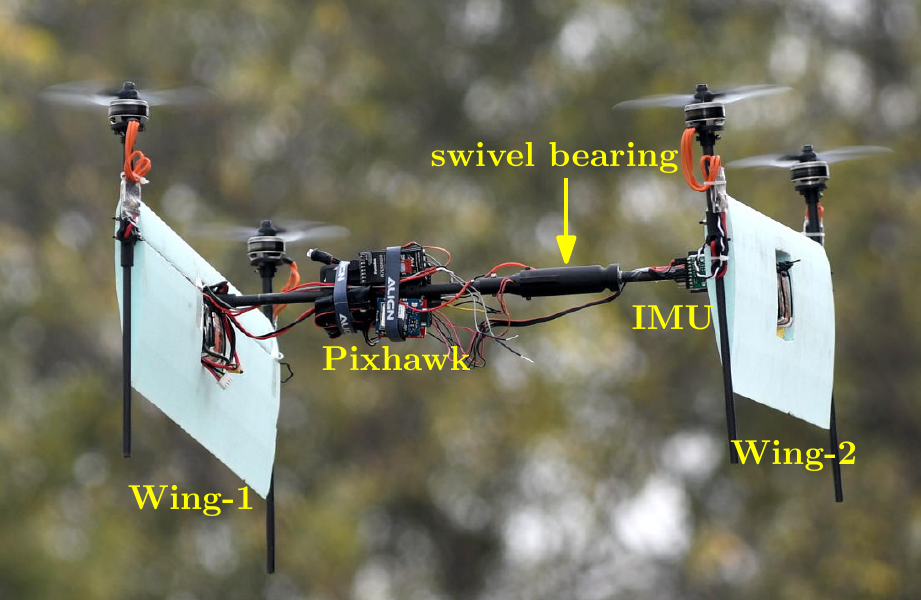}
\caption{Swiveling biplane quadrotor with the swiveling mechanism, autopilot and external IMU highlighted. \\ Experimental video link: \url{https://youtu.be/EXZ67fpw6I4}}
\label{fig:swivel_pic}
\end{center}
\end{figure}

\begin{figure}[!tbp]
\begin{center}
\includegraphics[width=1.0\linewidth]{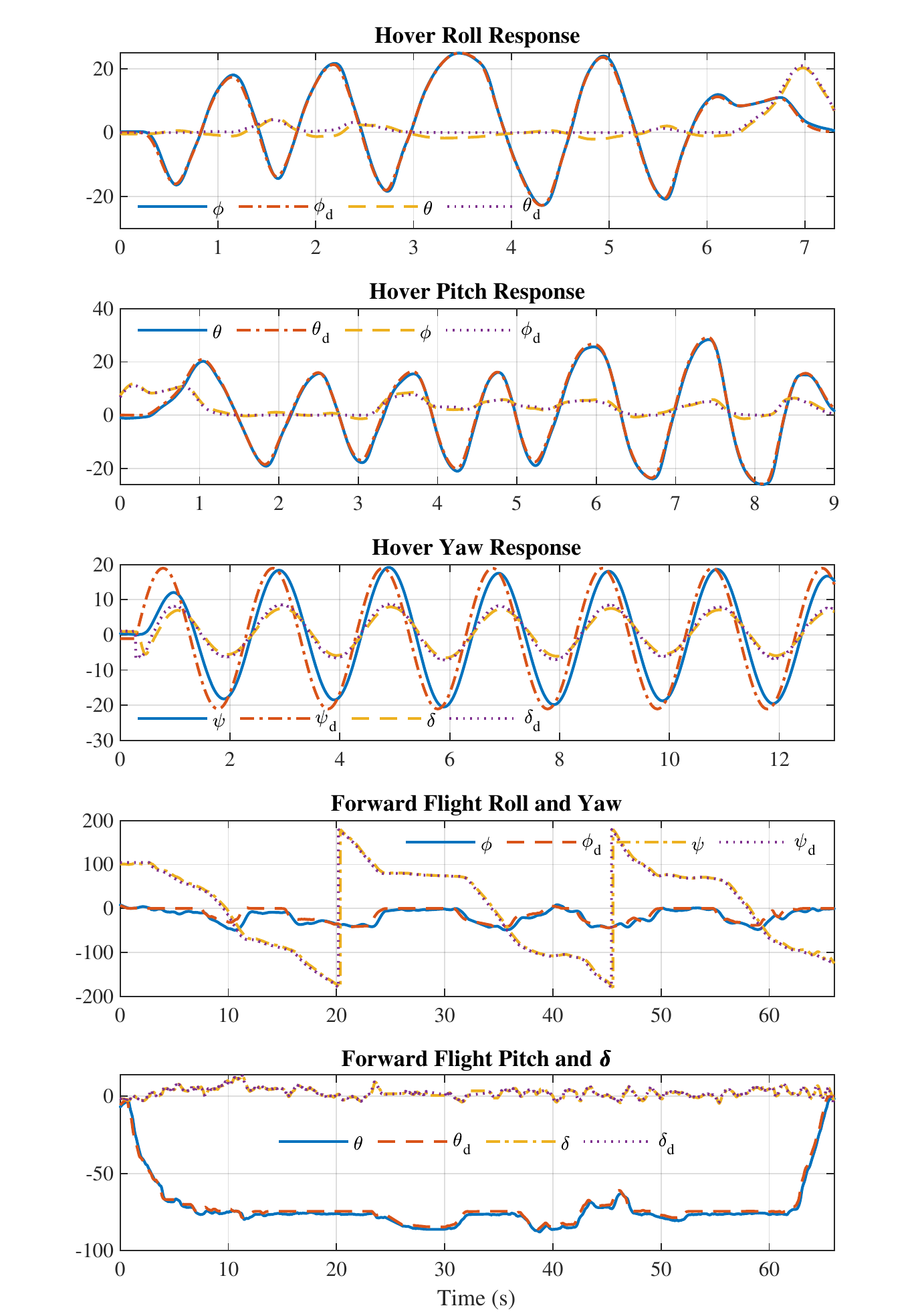}
\caption{First three plots show the tracking performance of the controller with roll, pitch and yaw axes excited individually. Roll and pitch input were given manually, while the yaw input was sinusoidal with 20 deg amplitude and a 2 second period. Last two plots show the tracking performance of the controller in forward flight. Roll and pitch input were given manually, while the yaw input was set to maintain zero sideslip.}
\label{fig:exp_forward}
\end{center}
\end{figure}

\section{Experimental Result} \label{sec:exp}
The experiments were conducted to evaluate the robustness of the proposed attitude tracking controller under near hover and forward flight conditions.  For near hover experiments, separate flights were conducted for each axes to restrict the resulting unwanted translational motion. The vehicle was manually commanded in roll and pitch axes whereas a sinusoidal reference command of 20 deg amplitude and 2 second time period was given for the yaw axis. For the forward flight experiments, the vehicle was flown manually starting from hover and then transitioned into cruise flight by gradually pitching the vehicle forward by 90 deg. 

\subsection{Vehicle Construction}
The vehicle used for experimentation is shown in Fig. \ref{fig:swivel_pic} and its  parameters are given in Table \ref{tab:biplane_params}. The wings are made of Expanded PolyPropylene (EPP) foam and the structure is built from 8 mm carbon fiber tubes. The power system involves two 1850 mAh, 3 cell lithium polymer batteries attached to the wings, 2400 KV brushless DC motors, and 30 A ESC. The swiveling mechanism is made of two ball bearings enclosed in a 3D printed casing. The inner shaft of the bearing is attached to a carbon fiber (CF) rod which is rigidly attached to Wing-2, and the outer casing of the bearing is attached to the CF rod rigidly attached to Wing-1. Enclosed withing the casing is a rotary encoder to measure the swivel angle. A Pixhawk is rigidly attached to the CF rod of Wing-1 and an external IMU is attached to Wing-2 to measure the swivel rate. The PX4 flight stack with its default quaternion based attitude estimator fuses the onboard IMU data to estimate the attitude of Wing-1 ($R_1, \omega_1$) and a complementary filter was used to fuse the encoder data and external IMU data to obtain $(\delta,\dot{\delta})$ at 250 Hz. The attitude tracking controller was implemented as a separate module and runs at 250 Hz. 

\subsection{Discussion}
The first set of experiments were for near hover conditions and the results are depicted in the top three plots of Fig. \ref{fig:exp_forward}. The first experiment involved manually exciting the vehicle about roll axis with the other two axes being held close to zero input. The second involved a similar manual excitation about the pitch axes. In both the aforementioned cases, the tacking is almost perfect with minor tracking deviation and negligible delay. Next, the yaw axis was commanded to track a sinusoidal command, and the relatively delayed response could be attributed to the tracking error in $\delta$ dynamics and disturbance torque due to aerodynamic interaction of propeller wake with wing. 

The last two plots of Fig. \ref{fig:exp_forward} show the performance of the controller during transition and forward flight. The same set of gains were used in the forward flight as the hover test flights. It is observed that there is some deviation in the roll tracking performance in forward flight when compared to yaw tracking in hover conditions, both corresponding to rotations about the underactuated body frame Z-axis. This error could be attributed to relatively large unaccounted aerodynamic damping torque from wing and disturbance torque due to components placed on the connecting rod resulting in aerodynamic asymmetry. However, the tracking was found to be satisfactory in pitch and yaw axes during the transition and forward flight.

\section{Conclusion and Future work} \label{sec:future}
The swiveling biplane-quadrotor introduced in this work has interesting  attitude dynamics with a configuration dependent inertia. Due to the underactuated nature of the system in the attitude manifold, an output tracking problem was posed in which the attitude of an intermediate virtual frame was chosen as the output. The method of dynamic feedback-linearization on a nominal model with constant inertia resulted in trivial zero dynamics and the resulting extended system had second order moment dynamics. The error due to cancellation of certain non-linearities and the constant inertia approximation was suppressed by the feedback component of the controller. The performance of the output tracking controller was found to be satisfactory both in simulation and manual flight experiments. The controller derived for the nominal model is also shown to have almost-global convergence property, which is the best that can be achieved for a system evolving on a non-Euclidean space \cite{bhat2000topological}.

As part of the future work we intend to incorporate the effect of aerodynamic forces on the vehicle and use this to augment the existing attitude tracking  controller for the biplane-mode cruise flight.


\nocite{*}
\bibliography{ref}
\end{document}